\documentclass[a4paper,DIV=14]{scrartcl}
\usepackage[utf8]{inputenc}

\newenvironment{keywords}%
{\begin{quotation} \noindent \textbf{Keywords:}}% oder „Keywords:”
{\end{quotation}}

\setkomafont{title}{\normalfont\bfseries}
\setkomafont{section}{\normalfont\bfseries\LARGE}
\setkomafont{subsection}{\normalfont\bfseries\Large}
\setkomafont{subparagraph}{\normalfont\bfseries}

\usepackage[backend=bibtex,style=alphabetic,citestyle=alphabetic,natbib=true,maxbibnames=99]{biblatex} % Use the bibtex backend with the authoryear citation style (which resembles APA)
\usepackage[autostyle=true]{csquotes} % Required to generate language-dependent quotes in the bibliography

\usepackage{amsmath}
\usepackage{amsthm}
\usepackage{amssymb}
\usepackage{microtype}%if unwanted, comment out or use option "draft"
\usepackage{stmaryrd}
\usepackage{complexity}
\usepackage{enumerate}

\theoremstyle{plain}
\newtheorem{theorem}{Theorem}[section]
\newtheorem{lemma}[theorem]{Lemma}
\newtheorem{fact}[theorem]{Fact}

\theoremstyle{definition}
\newtheorem{definition}[theorem]{Definition}

\usepackage{hyperref}

\newcommand{\N}{\ensuremath{\mathbb{N}}}

\newcommand{\set}[2]{\ensuremath{\left\{ #1 \mid #2 \right\}}}

\newcommand{\minorfree}[1]{\ensuremath{   \mathrm{MF}(#1)  }}

\newcommand{\tpall}{\ensuremath{   \mathbb{P}_{\mathrm{all}}  }}
\newcommand{\tpfin}{\ensuremath{   \mathbb{P}_{\mathrm{fin}}  }}

\newcommand{\amred}{\ensuremath{   \leq^{\ccex}_{\mathrm{m}}   }}

\newcommand{\pmred}{\ensuremath{   \leq^{\P}_{\mathrm{m}}   }}

\newcommand{\timecred}{\ensuremath{   \leq^{\TIME(n^c)}_{\mathrm{m}}   }}

\newcommand{\pmredu}{\ensuremath{   \leq^{\P}_{\mathrm{m},{\mathrm{u}}}   }}
\newcommand{\amredu}{\ensuremath{   \leq^{\ccex}_{\mathrm{m},{\mathrm{u}}}   }}

\newcommand{\mysufptred}{\ensuremath{   \leq^{\mathrm{fpt}'}   }}
\newcommand{\timecredu}{\ensuremath{   \leq^{\TIME(n^c)}_{\mathrm{m},\mathrm{u}}   }}
\newcommand{\timecredsu}{\ensuremath{   \leq^{\TIME(n^c)}_{\mathrm{m},\mathrm{su}}   }}

\newcommand{\parami}[1]{\ensuremath{ {\langle #1 \rangle}_{\mathrm{nu}} }}
\newcommand{\paramui}[1]{\ensuremath{ {\langle #1 \rangle}_{\mathrm{u}} }}

\newcommand{\paramsui}[1]{\ensuremath{ {\langle #1 \rangle}_{\mathrm{su}} }}

\newcommand{\ccex}{\ensuremath{\ComplexityFont{A}}}

\newcommand{\fptred}{\ensuremath{\leq^{\mathrm{fpt}}}}

\begin{document}

\title{Fundamentals of Parameterized Complexity Revisited} 
\author{Maurice Chandoo\footnote{Leibniz Universität Hannover,
        Institut für Theoretische Informatik,
        Appelstr.~4, 30167 Hannover, Germany; \hskip 2em      
        E-Mail: \href{mailto:chandoo@thi.uni-hannover.de}{chandoo@thi.uni-hannover.de} }}
\date{\vspace{-5ex}}

\maketitle

\begin{abstract}\noindent\textbf{Abstract.}
    Flum and Grohe define a parameter (parameterization) as a function $\kappa$ which maps words over a given alphabet to natural numbers. They require such functions to be polynomial-time computable. We show how this technical restriction can be lifted without breaking the theory. More specifically, instead of $\kappa$ we consider the set of languages that it bounds as parameterization and define the basic notions of parameterized complexity in terms of promise problems, which completely replace slices. One advantage of this formalization is that it becomes possible to interpret any complexity-theoretic concept which can be considered on a restricted set of inputs as a parameterized concept. Moreover, this formalization provides a unified way to apply the parameterization paradigm to other kinds of complexity such as enumeration or approximation by simply defining promise problems. 
\end{abstract}


\section{Introduction}

The purpose of parameterized complexity is to provide a different approach to deal with algorithmic problems that are deemed to be intractable in the classical sense. Instead of measuring the complexity of an algorithm solely in terms of the input length one considers other numerical parameters of the input instance as well. For example, the vertex cover problem is $\NP$-hard and therefore intractable in the classical sense. However, if one additionally measures the runtime in terms of the size of the vertex cover $k$ then it is possible to solve this problem in time $2^{\mathcal{O}(k)} \cdot n^{\mathcal{O}(1)}$. If $k$ is sufficiently small for the instances that one wants to solve, this can be considered tractable. This refined notion of efficiency is called fixed-parameter tractability and has become part of the canon in computational complexity.

There is the original formalization by Downey and Fellows which defines a parameterized problem as a subset of $\Sigma^* \times \N$ \cite{downey} and there is the formalization by Flum and Grohe which defines a parameterized problem as a tuple $(L,\kappa)$ where $L$ is a language and $\kappa$ is a polynomial-time computable function which maps instances to their parameter values \cite{flum}. Every formal parameterized problem $(L,\kappa)$ can be translated to $\set{(x,k)}{ x \in L , k = \kappa(x) }$. In the other direction, a subset $L$ of $\Sigma^* \times \N$ can be mapped to $(L',\kappa')$ where $L' = \set{ x\#1^k }{ (x,k) \in L }$ and $\kappa(y) = k$ if $y$ is of the form $x\#1^k$ and $1$ otherwise. Both translations preserve fixed-parameter tractability of the parameterized problem. Consequently, if one is primarily interested in whether a given parameterized problem is fixed-parameter tractable the difference between the two formalizations is immaterial. However, if one is interested in more basic questions such as ``what is a parameterized problem'' or ``what is a parameterized algorithm'' then it is worth to examine these two formalizations more closely. 

The crucial difference between these two formalizations is the `position' of the parameter value. 
In Downey and Fellows' formalization it is part of the input whereas in the case of  Flum and Grohe's formalization it is computed from the input. However, Flum and Grohe have to make the additional requirement that computing the parameter value must be possible in polynomial-time; without it their theory of FPT would `break' for reasons that we explain in Section~\ref{ss:fg}. Additionally, some commonly considered parameterizations such as tree-width or clique-width are $\NP$-hard and thus violate this requirement. They mention that such parameterizations can be considered nonetheless by making the parameter value part of the input. Consistently applying this `quick fix' leads back to Downey and Fellows' formalization. 
We show that it is possible to remove the polynomial-time computability restriction without breaking the theory. In our formalization the parameter value is neither part of the input nor computed from it; it is provided as promise in a certain sense.

\subparagraph{Overview of the paper.} 
The two main observations on which our formalization is based are the following. First, we consider a parameterized algorithm to be a sequence of algorithms $A_1,A_2,\dots$ where $A_k$ solves all instances whose parameter value is at most $k$. Secondly, for a function $\kappa \colon \Sigma^* \rightarrow \N$ we say a language $L$ over $\Sigma$ is bounded by $\kappa$ if there exists a $k \in \N$ such that $\kappa(w) \leq k$ holds for all $w \in L$. Let $\mathbb{K}(\kappa)$ denote the set of languages bounded by $\kappa$. Instead of $(L,\kappa)$ we consider $(L,\mathbb{K}(\kappa))$ as a parameterized problem. Intuitively, we see $(L,\mathbb{K}(\kappa))$ as the set of promise problems $(L,P)$ with $P \in \mathbb{K}(\kappa)$. This choice is justified as follows.
Observe that the concrete parameter values do not influence the complexity of a parameterized problem, e.g.~$(L,\kappa)$ and $(L, f(\kappa(x)))$  with $f(n) = 2^n$ have the same complexity. This can be interpreted as the fact that a parameter $\kappa$ is just a representation of a higher-level object which actually determines the parameterized complexity: its normalized inverse $\mathbb{K}(\kappa)$. Consequently, one might argue that complexity classes and reductions should be defined in terms of $\mathbb{K}(\kappa)$ rather than $\kappa$ for the same reason that graph theory is defined in terms of graphs as opposed to adjacency matrices.
More formally, let $\kappa,\tau$ be parameters over some alphabet $\Sigma$. We show that if $\kappa$ and $\tau$ bound the same set of languages, i.e.~$\mathbb{K}(\kappa) = \mathbb{K}(\tau)$, then the parameterized complexity of $(L,\kappa)$ and $(L,\tau)$ is identical for any language $L$ (see Section~\ref{sec:connection}).

In Section~\ref{sec:crit} we recapitulate the existing definitions and raise some theoretical issues that we see with them.
In Section~\ref{sec:nar} we show how to define parameterized complexity classes and reductions in terms of $(L,\mathbb{K}(\kappa))$. While the definitions provided here are given for decision problems, they easily translate to other kinds of problems. Moreover, slices become superfluous since they are replaced by promise problems. In the beginning of the section we provide a narrative which leads up to this formalization. In Section~\ref{sec:connection} we show how our definitions formally relate to the ones given by Flum and Grohe; their definitions are a special case of ours. In Section~\ref{sec:final} we summarize our perspective on parameterized complexity and the rationale that led us to the definitions presented here. For additional motivation, it might be advisable to skim over the last section before starting with Section~\ref{sec:nar}.

\section{Criticism of Existing Formalizations}
\label{sec:crit}
\subsection{Downey and Fellows}
\label{ss:df}
Downey and Fellows define a parameterized problem $L$ as a subset of $\Sigma^* \times \Sigma^*$ for an alphabet $\Sigma$. They remark that the parameter value (the second component) need not be a numerical value in principal. For example, it is reasonable that an algorithm which solves a problem parameterized by tree-width expects a tree-decomposition as additional input. However, they also say that it can be assumed w.l.o.g.~that the parameter value is numerical, i.e.~$L \subseteq \Sigma^* \times \N$. A parameterized problem $L$ is called fixed-parameter tractable if there exists an algorithm $A$, a constant $c \in \N$ and a computable function $f \colon \N \rightarrow \N$ such that for all $x \in \Sigma^*$ and $k \in \N$ the algorithm $A$ accepts $(x,k)$ iff $(x,k) \in L$ and $A$ runs in time $f(k)\cdot n^c$ with $n = |x|$.

In natural language a parameterized problem is commonly expressed as $X$ parameterized by $Y$. For example, the vertex cover problem parameterized by solution size or the Hamiltonian cycle problem parameterized by tree-width. Formalizing these two examples requires some further thought. For instance, the vertex cover problem parameterized by solution size can either be formalized as 
\begin{itemize}
    \item $\set{(G,k) \in \Sigma^* \times \N }{ G \text{ has a vertex cover of size } k} $, or
    \item $\set{(G,k,k') \in \Sigma^* \times \N \times \N }{ G \text{ has a vertex cover of size } k \text{ and } k' = k} $.
\end{itemize}
The second example might seem contrived at first but the rationale behind it becomes apparent when we consider how to formalize the second problem:
$$\set{(G,k) \in \Sigma^* \times \N }{G \text{ is Hamiltonian and twd($G$)} = k}$$
Generally speaking, we need a mechanism to translate a decision problem along with a parameterization to a parameterized problem. This mechanism should be independent of the decision problem that we want to translate. In the case of the Hamiltonian cycle problem a numerical component (the parameter) is added to the original decision problem, which is functionally determined by the original instance $G$. This is consistent with the second formalization of the vertex cover problem but not with the first one. In the case of the first formalization the numerical part of the original instance is taken away and casted as parameter value. This is based on the coincidence that the parameter value is already part of the original input and thus arbitrary.  

Every instance $G$ of the Hamiltonian cycle problem is associated with a parameter value, i.e.~the tree-width of $G$. Some of this information is lost in the formalized version. Consider a graph parameter $\lambda$ such that $\lambda(G)$ is the tree-width of $G$ whenever $G$ is Hamiltonian and otherwise $\lambda(G)$ is an arbitrary natural number. The formalized version of the Hamiltonian cycle problem parameterized by $\lambda$ is identical to the one parameterized by tree-width. This implies that the parameter value of negative instances (in this case non-Hamiltonian graphs) is irrelevant in parameterized complexity. This might be true as long as one only considers problems in $\NP$. The common approach to solve such a problem is to find a witness for the given instance. The parameter value (partially) determines how much time the algorithm is granted to find such a witness. Since negative instances do not posses a witness the parameter value is irrelevant. This implication becomes questionable when considering problems beyond $\NP$ in the parameterized context. For example, in \cite[Thm.~3]{ordyniak} an fpt-algorithm for QBF parameterized by dependency tree-width (a novel parameter) is given. If the given QBF is false then the algorithm outputs a (Q-resolution) refutation of size $\mathcal{O}(3^k n)$.  This indicates that the parameter value of negative instances can carry meaningful information which should not be discarded. However, this seems to be irreconcilable with the definition of a parameterized problem as a subset of $\Sigma^* \times \N$.

Another issue occurs when considering closure under complement. Assume that a graph property $X$ parameterized by some graph parameter $\lambda$ is fixed-parameter tractable via an algorithm $A$. This means $A$ accepts the input $(G,k)$ iff $G \in X$ and $\lambda(G) = k$. 
Let $\mathcal{G}$ denote the set of all graphs. 
Intuitively, it should hold that the complement $\overline{X} = \mathcal{G} \setminus X$ parameterized by $\lambda$ is fixed-parameter tractable as well by taking $A$ and flipping its answers. However, if one flips the answers of $A$ then it recognizes $\set{(G,k)}{ G \notin X \vee \lambda(G) \neq k }$  which is different from $\set{(G,k)}{ G \notin X \wedge \lambda(G) = k }$. This brings us to the next issue: what is the intuitive parameterized problem that a set such as $\set{(G,k)}{ G \notin X \vee \lambda(G) \neq k }$ or $\set{(x,k)}{x \in L \wedge k \text{ is prime}}$ for some language $L$ represents? Is it possible to sensibly interpret every subset of $\Sigma^* \times \N$ as parameterized problem or does this definition of a parameterized problem over-approximate the intuitive notion? 

\subsection{Flum and Grohe}
\label{ss:fg}
A parameter over an alphabet $\Sigma$ is a total, polynomial-time computable function $\kappa \colon \Sigma^* \rightarrow \N$ (they call $\kappa$ parameterization but we reserve this term for something else). 
A parameterized problem is a tuple $(L,\kappa)$ where $L$ is a language and $\kappa$ is a parameter, both over the same alphabet. A parameterized problem $(L,\kappa)$ over $\Sigma$ is fixed-parameter tractable (in $\FPT$) if there exists an algorithm $A$, a constant $c \in \N$ and a computable function $f \colon \N \rightarrow \N$ such that for all $x \in \Sigma^*$ it holds that $A$ accepts $x$ iff $x \in L$ and $A$ runs in time $f(\kappa(x)) \cdot n^c$ with $n = |x|$. The $i$-th slice of a parameterized problem $(L,\kappa)$ is defined as $(L,\kappa)_i := \set{x \in L}{ \kappa(x) = i}$ for $i \in \N$. A parameterized problem $(L,\kappa)$ is in $\XP_{\mathrm{nu}}$ if $(L,\kappa)_i$ is in $\P$ for all $i \in \N$. 

This definition matches the natural language description of a parameterized problem. However, the requirement that $\kappa$ needs to be polynomial-time computable defies intuition. For instance, the straightforward formalization of the Hamiltonian cycle problem parameterized by tree-width would be $(\set{G \in \mathcal{G}}{\text{$G$ is Hamiltonian}},\textrm{twd})$. But this formalization is illegitimate since tree-width is not polynomial-time computable unless $\P = \NP$. Flum and Grohe remark that this is a technical requirement and non-polynomial-time computable parameters can be considered nonetheless by modifying the original problem such that the parameter value becomes part of the input. As stated previously, consistently applying this quick fix---rather than arbitrarily based on the complexity of the parameter---essentially leads back to Downey and Fellows' formalization. 

To understand the significance of this requirement let us consider what happens if it is dropped. One consequence is that the parameter value $\kappa(x)$ of an instance $x$ cannot be computed by an fpt-algorithm (or even XP-algorithm) anymore and thus becomes inaccessible whenever $\kappa$ is not polynomial-time computable. 
Another consequence is that $\FPT \subseteq \XP_{\mathrm{nu}}$ does not longer hold. 
Let $A$ be a language over $\Sigma$ which is not in $\P$ and for $x \in \Sigma^*$ let $\kappa(x) = 1$ if $x \in A$ and $2$ otherwise. It holds that $(\Sigma^*,\kappa)_1 = A $ is not in $\P$ and thus $(\Sigma^*,\kappa)$ is not in $\XP_{\mathrm{nu}}$. The problem is trivially in $\FPT$ via the algorithm which always accepts. 

The notion of a parameterized problem ontologically precedes complexity classes. Adjusting its definition in order to resolve issues with theoretically arbitrary complexity classes  ($\FPT$ and $\XP_{\mathrm{nu}}$) seems questionable. If parameterized complexity had focused on space instead of time efficiency from the beginning, an analogous train of thought would have led one to conclude that a parameter should be logspace computable (the logspace counterparts of $\FPT$ and $\XP_{\mathrm{nu}}$ exhibit the same two issues described above). The definition of what a parameterized problem is, should not change depending on the complexity classes under consideration. However, this does happen: for instance, in \cite{elberfeld} a parameter is required to be first-order computable because the authors consider small circuit complexity classes. 

The following critique concerning slices applies to both formalizations.
An issue with slices is that they conflate problem and parameter complexity. Consider the problem $(\Sigma^*,\kappa)$ used to show that $\FPT \not\subseteq \XP_{\mathrm{nu}}$. Despite the fact that its problem part $\Sigma^*$ is trivial it is not in $\XP_{\mathrm{nu}}$, a class that should intuitively be a superset of $\FPT$. In this case the complexity of the parameterized problem is dominated by the complexity of the parameter. Should the parameter complexity influence the complexity of a parameterized problem? 
If one is convinced that the problem $(\Sigma^*,\kappa)$ is of trivial complexity regardless of whatever $\kappa$ might be then the answer is no. This conflation also occurs whenever a parameterized algorithm needs to compute $\kappa(x)$. 

\section{From Classical To Parameterized Complexity}
\label{sec:nar} 
Imagine you want to design a polynomial-time algorithm for the graph isomorphism problem. After some failed attempts you settle for trying to solve this problem for restricted graph classes. After a while, you find a polynomial-time algorithm $A$ which solves isomorphism for planar graphs. Since planar graphs can be characterized as $(K_{3,3},K_{5})$-minor free graphs, you try to analyze your algorithm $A$ from this perspective. You notice that the algorithm $A$ relies on the fact that the input graphs are $K_{3,3}$-minor free at only one particular step. After a while, you figure out a way to modify $A$ such that it does not rely on this assumption anymore and call the new algorithm $A'$. Stated differently, $A'$ is a polynomial-time isomorphism test for $K_5$-minor free graphs. Naturally, you wonder whether there is anything special about the number 5 in $K_5$ or whether $A'$ can be modified such that it works for $K_6$-minor free graphs as well. Eventually, you manage to show that for every $i \in \N$ there is a polynomial-time algorithm $A_i$ which solves isomorphism for $K_i$-minor free graphs. 

We call a sequence of algorithms $A_1,A_2,\dots$ with the following properties a parameterized algorithm. First, each algorithm is efficient (e.g.~runs in polynomial time). Secondly, we want this sequence of algorithms to be monotone in the sense that each subsequent algorithm solves more and more inputs correctly. Thirdly, to prevent that there might be inputs for which no algorithm works, we require that for every input $x$ there exists an $i \in \N$ such that $A_i$ correctly solves $x$. In a sense one tries to approximate an efficient algorithm. The sequence of algorithms from the previous paragraph satisfies all of these requirements and, hence, is a parameterized algorithm for the graph isomorphism problem.

Seeing a parameterized algorithm as a sequence of algorithms is not new at all. For example, in \cite{bltw} the main result is stated as: for every $k \in \N$ there exists a linear-time algorithm $A_k$ that decides whether a given graph $G$ has tree-width at most $k$, and if so, outputs a tree-decomposition of $G$ with tree-width at most $k$. 

It is trivial to find a parameterized algorithm for any problem $X$: let $A_i$ be the algorithm which correctly solves all inputs up to length $i$ by using a look-up table. The other extreme is that we have a single efficient algorithm $A$ which solves $X$ for all inputs and thus we can set $A_i = A$ for all $i \in \N$, which makes this approach superfluous. In order to distinguish such trivial parameterized algorithms from interesting ones we need to formalize what exactly it is that is solved by a parameterized algorithm.
 
The problem $X$ restricted to inputs from $P$ (the promise) is called promise problem $(X,P)$. We say an algorithm $A$ solves $(X,P)$ if it correctly (w.r.t.~$X$) solves every input $x$ from $P$.
A parameterized algorithm $A_1,A_2,\dots$ solves a sequence of promise problems $(X,P_1),(X,P_2),\dots$. 
The second property of a parameterized algorithm implies that $P_1 \subseteq P_{2} \subseteq \dots$. Let us call the sequence $(P_i)_{i \in \N}$ a parameterization. 
Consider a problem $X$ and two parameterizations $\bar{P} = (P_i)_{i \in \N}, \bar{P'} = (P'_i)_{i \in \N}$.
We say $\bar{P}$ and $\bar{P'}$ are equivalent if for all $i \in \N$ there exists a $j \in \N$ such that $P_i \subseteq P'_j$ and vice versa ($\forall i \: \exists j : P'_i \subseteq P_j$). 
If $\bar{P}$ and $\bar{P'}$ are equivalent then a parameterized algorithm $(A_i)_{i \in \N}$ solves $(X,\bar{P})$ iff it solves $(X,\bar{P'})$ (up to reindexing of
the parameterized algorithm). For instance, suppose that $A_i$ solves $(X,P_i)$ for all $i \in \N$. Since for every $i \in \N$ there exists a $j \in \N$ such that $P'_i \subseteq P_j$ it follows that $(X,P'_i)$ is solved by $A_j$. 
Hence, if one's intention is to prove (or disprove) the existence of a parameterized algorithm for a particular problem the choice between two equivalent parameterizations is irrelevant.  
To eliminate this ambiguity we consider the normalization $\set{P}{ \exists i \in \N : P \subseteq P_i}$ of a parameterization $\bar{P}$. Two parameterizations are equivalent iff their normalizations coincide. From now on we use the term parameterization to denote the normalized entity.

\subsection{Parameterizations}

\begin{definition}
    Let $\mathbb{P}$ be a set of languages over an alphabet $\Sigma$. We call $\mathbb{P}$ parameterization if:
    \begin{enumerate}
        \item $\mathbb{P}$ is closed under subsets and finite union
        \item $\mathbb{P}$ contains the singleton language $\{x\}$ for all $x \in \Sigma^*$
        \item there exists a countable subset $\{ P_1, P_2, \dots \}$ of $\mathbb{P}$ such that for all $P \in \mathbb{P}$ there exists an $i \in \N$ with $P \subseteq P_i$
    \end{enumerate}    
\end{definition}

Let us make some remarks about the structure of parameterizations.
We call the set of all languages ($\tpall$) and the set of all finite languages ($\tpfin$) trivial parameterizations. For two sets of languages $\mathbb{A},\mathbb{B}$ over an alphabet $\Sigma$ let $\mathbb{A} \,\sqcap\, \mathbb{B} := \set{ C \subseteq \Sigma^* }{ \exists A \in \mathbb{A}, B \in \mathbb{B} : C = A \cap B }$. We define $\mathbb{A} \,\sqcup\, \mathbb{B}$ analogously, i.e.~replace `$\cap$' with `$\cup$' in the previous sentence. 
If $\mathbb{P}$ and $\mathbb{P}'$ are parameterizations then $\mathbb{P} \,\sqcap\, \mathbb{P}'$ and $\mathbb{P} \,\sqcup\, \mathbb{P}'$ are parameterizations as well. Also,  $\tpfin \subseteq \mathbb{P} \subseteq \tpall$ holds for all parameterizations $\mathbb{P}$. Therefore the set of parameterizations forms a bounded lattice. Due to closure under subsets it follows that $\mathbb{P} \,\sqcap\, \mathbb{P}' = \mathbb{P} \cap \mathbb{P}'$. For a language $L$ over $\Sigma$ let $\overline{L}$ denote its complement $\Sigma^* \setminus L$. If $\mathbb{P}$ is a parameterization then it holds that $\set{\overline{P}}{ P \in \mathbb{P}}$ is a parameterization iff $\mathbb{P} = \tpall$. 

\begin{definition}
    A parameterized problem over an alphabet $\Sigma$ is a tuple $(L,\mathbb{P})$ where $L$ is a language and $\mathbb{P}$ is a parameterization, both over $\Sigma$.     
\end{definition}

The definition of  parameterization can be directly applied to any countable set of objects such as graphs or integers. In the following we give four examples of parameterizations for the set of all unlabeled graphs $\mathcal{G}$.

For a graph $H$ we say a graph class $\mathcal{C}$ (a set of unlabeled graphs) is $H$-minor free if no graph in $\mathcal{C}$ contains $H$ as minor. 
 For a graph class $\mathcal{H}$ let $\minorfree{\mathcal{H}}$ denote the set of graph classes that are $H$-minor free for some $H$ in $\mathcal{H}$. 
The set of graph classes $\minorfree{\mathcal{G}}$ forms a parameterization. 
For closure under union consider two graph classes $\mathcal{C},\mathcal{D}$ in  $\minorfree{\mathcal{G}}$ such that $\mathcal{C}$ is $H$-minor free and $\mathcal{D}$ is $H'$-minor free for some graphs $H,H'$. It holds that $\mathcal{C} \cup \mathcal{D}$ is $(H'')$-minor free where $H''$ denotes the disjoint union of $H$ and $H'$. Secondly, each singleton graph class $\{G\}$ is $H$-minor free where $H$ is an arbitrary graph that has more vertices than $G$.
To show that the third condition holds let $f \colon \N \rightarrow \mathcal{G}$ be a bijective function. 
The required countable subset of $\minorfree{\mathcal{G}}$ is $\mathcal{C}_1,\mathcal{C}_2,\dots$ where $\mathcal{C}_i$ is the set of $f(i)$-minor free graphs. This parameterization is also known as the set of proper minor-closed graph classes closed under subsets.
A notable subset of $\minorfree{\mathcal{G}}$ is $\minorfree{\mathrm{Planar}}$, which is also a parameterization. The set $\minorfree{\mathrm{Planar}}$ is probably more commonly known as the set of graph classes with bounded tree-width; this equality is proved in \cite{robertson}.

The following parameterization admits various characterizations.
A graph class $\mathcal{C}$ is called sparse if there exists a $c \in \N$ such that every graph $G$ in $\mathcal{C}$ has at most $cn$ edges where $n$ is the number of vertices in $G$. A graph class $\mathcal{C}$ is called uniformly sparse if there exists a $c \in \N$ such that every graph $G$ which occurs as induced subgraph of some graph in $\mathcal{C}$ has at most $cn$ edges. A graph class is called hereditary if it is closed under vertex deletion. 
The arboricity of a graph $G$ is the least $k \in \N$ such that there exist $k$ forests $F_1,\dots,F_k$ with the same vertex set as $G$ and $E(G) = \cup_{i=1}^k E(F_i)$.  The degeneracy of a graph $G$ is the least $k \in \N$ such that every induced subgraph $H$ of $G$ has a vertex with degree at most $k$. The following equalities are well-known and not difficult to prove:
$$ \text{Uniformly Sparse} = [\mathrm{Sparse} \cap \mathrm{Hereditary}]_{\subseteq} = \text{bounded degeneracy} = \text{bounded arboricity} $$
where $[\cdot]_{\subseteq}$ denotes closure under subsets. Every proper minor-closed graph class is sparse and hereditary. 

The set of graph classes which do not contain the complete bipartite graph $K_{i,j}$ as (not necessarily induced) subgraph for some $i,j \in \N$ forms a parameterization. This parameterization has been considered for the dominating set problem in \cite{philip}. Every graph with degeneracy at most $d$ is $K_{d+1,d+1}$-free. The parameterizations are related as follows:
$$ \minorfree{\text{Planar}} \subseteq \minorfree{\mathcal{G}} \subseteq \text{Uniformly Sparse} \subseteq K_{i,j}\text{-free}  $$ 

The subset relation measures the difficulty of parameterizations. 
Let $X$ be a graph property such as being Hamiltonian. The parameterized problem $(X,\minorfree{\mathcal{G}})$ is intuitively at least as hard as $(X,\minorfree{\text{Planar}})$. 

Many decision problems include a numerical component which specifies the solution size, e.g.~the clique problem can be defined as $\mathrm{Clique} = \set{ (G,k) \in \mathcal{G} \times \N}{ G \text{ has a $k$-clique} }$. A parameterization of this problem consists of two parts: a parameterization $\mathbb{P}_1$ of $\mathcal{G}$ and $\mathbb{P}_2$ of $\mathbb{N}$. The first part $\mathbb{P}_1$ is called structural parameterization. The parameterized clique problem is represented by $(\mathrm{Clique}, \tpall \times \tpfin)$ in our formalism; here $\tpall$ means the set of all graph classes and $\tpfin$ means the set of all finite subsets of $\mathbb{N}$.

\subsection{Complexity Classes}

\begin{definition} A promise problem over an alphabet $\Sigma$ is a tuple $(L,P)$ where $L$ and $P$ are languages over $\Sigma$ and $P$ is called promise. Let $\ccex$ be a set of languages. We say $(L,P)$ is in $\ccex$ if there exists a language $L'$ in $\ccex$ such that $L \cap P = L' \cap P$.    
\end{definition}

The condition $L \cap P = L' \cap P$ can be equivalently stated as $x \in L \Leftrightarrow x \in L'$ for all $x \in P$. The latter formulation makes it clear that a promise problem only has to be correctly solved on inputs from its promise. We use the former formulation due to its brevity. 

\begin{definition}[Non-Uniform]
    Let $\ccex$ be a set of languages. The set of parameterized problems $\parami{\ccex}$ is defined as follows. A parameterized problem $(L,\mathbb{P})$ is in $\parami{\ccex}$ if the promise problem $(L,P)$ is in $\ccex$ for all $P \in \mathbb{P}$. 
\end{definition}

The parameterized interpretation $\parami{\ccex}$ of a set of languages $\ccex$ is a generalization in the sense that $L$ is in $\ccex$ iff $(L,\tpall)$ is in $\parami{\ccex}$. Therefore no information is lost in this interpretation. Observe that unlike slices (the de facto standard for defining non-uniform parameterized complexity classes) the concept of a promise problem, i.e.~solving a problem on a restricted set of inputs, naturally translates to all kinds of complexity such as approximation, enumeration or counting. Consequently, it is straightforward to define non-uniform parameterized complexity classes for them.

We shall use complexity class as an informal term in the following sense. We call a countable set of algorithms a complexity class. A decision algorithm $M$ decides the language $L(M)$. We also use the name of the complexity class to refer to the set of languages that are decided by one of its decision algorithms. For instance, $\P$ is the set of deterministic Turing machines (DTMs) which run in polynomial time.

\begin{definition}
	Let $\Sigma$ be an alphabet. We call a computable, total function $\kappa \colon \Sigma^* \rightarrow \N$ a parameter. 
    For $c \in \N$ let $\kappa_c$ denote the set $\set{x \in \Sigma^*}{ \kappa(x) \leq c}$. A language $L$ over $\Sigma$ is bounded by $\kappa$ if there exists a $c \in \N$ such that $L \subseteq \kappa_c$. We write $\mathbb{K}(\kappa)$ to denote the set of languages bounded by $\kappa$.    
\end{definition}

\begin{definition}[Uniform]
	Let $\ccex$ be a complexity class. The set of parameterized problems $\paramui{\ccex}$ is defined as follows.    
	A parameterized problem $(L,\mathbb{P})$  over an alphabet $\Sigma$ is in $\paramui{\ccex}$ if there exists a computable sequence of algorithms $(M_k)_{k \in \N}$ from $\ccex$ and a parameter $\kappa$ over $\Sigma$ such that $\mathbb{P} \subseteq \mathbb{K}(\kappa)$ and $L \cap \kappa_i = L(M_i) \cap \kappa_i$ for all $i \in \N$. 
	\label{def:upc}
\end{definition}

The condition $\mathbb{P} \subseteq \mathbb{K}(\kappa)$ ensures that for every $P \in \mathbb{P}$ there exists an $i \in \N$ such that $P \subseteq \kappa_i$ and thus $(L,P)$ is solved by $M_i$. Therefore $\paramui{\ccex} \subseteq \parami{\ccex}$ holds for every complexity class $\ccex$. The parameter $\kappa$ can be seen as algorithm selector which tells us what algorithm to use for what instance, i.e.~in order to solve instance $x$ the algorithm $M_{\kappa(x)}$ can be used.

Observe that $(M_k)_{k \in \N}$ satisfies the three conditions for a parameterized algorithm that we have postulated in the beginning. If we consider algorithms from $\ccex$ to be efficient then every algorithm $M_k$ with $k \in \N$ is efficient since it is from $\ccex$. 
Monotonicity is reflected by the fact that there exists a sequence of promises $P_1 \subseteq P_2 \subseteq \dots$ which characterizes $\mathbb{P}$ such that $(L,P_i)$ is solved by $M_i$. This kind of monotonicity is defined in terms of $L$ and $\mathbb{P}$. There is another kind of monotonicity defined solely in terms of $L$: we call $(M_k)_{k \in \N}$ strongly monotone (w.r.t.~$L$) if  for all instances $x$ and $i \geq 2$ it holds that $M_{i}$ solves input $x$ correctly w.r.t.~$L$ whenever $M_{i-1}$ does. We remark that monotonicity is essential for the notion of uniformity because without it one could always (as in for every problem) choose $(M_k)_{k \in \N}$ such that it contains every algorithm from $\ccex$. If the requirement that the sequence of algorithms  $(M_k)_{k \in \N}$ and $\kappa$ must be computable is dropped then the resulting definition coincides with the non-uniform one.

\begin{lemma}
	Let $\ccex$ be a complexity class. The parameterized problem $(L,\mathbb{P})$ is in $\paramui{\ccex}$ if $L$ is decidable and there exists a computable sequence of algorithms $(M_k)_{k \in \N}$ from $\ccex$ such that:
	\begin{enumerate}
		\item for all $P \in \mathbb{P}$ there exists an $i \in \N$ such that $L \cap P = L(M_i) \cap P$		
		\item $(M_k)_{k \in \N}$ is strongly monotone w.r.t.~$L$
		\item for all $i \in \N$ the algorithm $M_i$ eventually halts for all inputs (is total)
	\end{enumerate}
    \label{lem:usc}
\end{lemma}
\begin{proof}
	We define a parameter $\kappa$ as follows. Let $\kappa(x)$ be the smallest $k \in \N$ such that $x \in L \Leftrightarrow x \in L(M_k)$. That $\kappa$ is a parameter (computable and total) follows from $L$ being decidable and the first and third condition. 
	We claim that $(L,\mathbb{P})$ is in $\paramui{\ccex}$ via $(M_k)_{k \in \N}$ and $\kappa$. It can be shown inductively  that $L \cap \kappa_i = L(M_i) \cap \kappa_i$ holds for all $i \in \N$ by using the second condition. It remains to argue that $\mathbb{P} \subseteq \mathbb{K}(\kappa)$. Let $P \in \mathbb{P}$. There exists an $i \in \N$ s.t.~$L \cap P = L(M_i) \cap P$. We show that $P \subseteq \kappa_i$. Let $x \in P$. It holds that $x \in L \Leftrightarrow x \in L(M_i)$. Therefore $\kappa(x) \leq i$ which implies $x \in \kappa_i$. 
\end{proof}

The parameter $\kappa$ defined in the previous proof from $L$ and $(M_k)_{k \in \N}$ induces an inclusion-maximal parameterization $\mathbb{K}(\kappa)$ in the sense that for every $(L,\mathbb{P})$ which is solved by $(M_k)_{k \in \N}$ it holds that $\mathbb{P} \subseteq \mathbb{K}(\kappa)$. Therefore if a parameterized algorithm is strongly uniform w.r.t.~$L$ it induces a unique maximal parameterization for which it solves $L$.

When trying to solve a problem from $\NP$ one often does so by computing a witness. In this case it can be assumed that an algorithm will never incorrectly output `yes'. If a sequence of algorithms $(M_k)_{k \in \N}$ has this property then it can be easily modified to ensure that it is strongly uniform. Let $(M'_k)_{k \in \N}$ be defined as follows. Let $M'_1 = M_1$ and let $M'_{i+1}$ be the algorithm which simulates $M'_i$ and $M_{i+1}$ and outputs `yes' iff one of the two simulations output `yes'. If every algorithm from $(M_k)_{k \in \N}$ runs in polynomial time then so does $(M'_k)_{k \in \N}$. 
This trick (presumably) cannot be applied to parameterized algorithms which solve problems such as QBF and in such cases being strongly monotone is a non-trivial property.    

The set of decidable languages $\R$ coincides with $\paramui{\R}$. This is in accordance with the intuition that the parameterized approach is only sensible in the context of efficiency. 
\begin{fact}
	Let $(L,\mathbb{P})$ be a parameterized problem.
	It holds that $L$ is in $\R$ iff $(L,\mathbb{P})$ is in $\paramui{\R}$.
\end{fact}

It follows that $\paramui{\P} \subsetneq \parami{\P}$.
Let $L$ be an undecidable language. The problem $(L,\tpfin)$ is in $\parami{\P}$ but not in $\paramui{\P}$ since this would imply that $L$ is decidable due to the previous fact.

There is another type of uniformity referred to as strongly uniform. 
It has the additional requirement that not only the sequence of algorithms $(M_k)_{k \in \N}$ has to be computable but also an upper bound on their resource usage. For each algorithm $M_i$ this upper bound is described by a  constant $c_i \in \N$. This means the function $f(i) = c_i$ is computable and the resource usage of $M_k$ is upper bounded by a function $r$ in terms of $f(k)$ and the input length. For example, for the class $\XP$ the function $r$ is $(f(k),n) \mapsto n^{f(k)}$. 

\begin{definition}[Strongly Uniform]
	Let $\ccex$ be a complexity class and let $r \colon \N \times \N \rightarrow \N$ be a computable, total function. The set of parameterized problems $\paramsui{\ccex,r}$ is defined as follows.   
	A parameterized problem $(L,\mathbb{P})$ over an alphabet $\Sigma$ is in $\paramsui{\ccex,r}$ if there exists a computable sequence of algorithms $(M_k)_{k \in \N}$ from $\ccex$, a parameter $\kappa$ over $\Sigma$ and a computable, total function $f \colon \N \rightarrow \N$ such that:
	\begin{enumerate}
		\item $(L,\mathbb{P})$ is in $\paramui{\ccex}$ via $(M_k)_{k \in \N}, \kappa$		
		\item for all $k \in \N$ and $x \in \Sigma^*$  it holds that $M_k$ on $x$ runs at most $r(f(k),|x|)$ steps
	\end{enumerate}
\end{definition}

We define $\paramsui{\P}$ as $\paramsui{\P,(f,n) \mapsto n^f}$ and $\paramsui{\TIME(n^c)}$ as $\paramsui{\TIME(n^c),(f,n) \mapsto f\cdot n^c}$ for $c \in \N$. Our counterpart of (strongly uniform) $\FPT$ is $\FPT' = \bigcup_{c \in \N} \paramsui{\TIME(n^c)}$.
Let us consider how existing results fit these definitions.

For every proper minor-closed graph class $\mathcal{C}$ the graph isomorphism problem GI on $\mathcal{C}$ can be decided in polynomial time \cite{ponomarenko}. Let $\mathrm{GI}$ be defined as $\set{(G,H) \in \mathcal{G} \times \mathcal{G}}{ G \cong H }$. Formally, this means $(\mathrm{GI},\minorfree{\mathcal{G}} \times \minorfree{\mathcal{G}})$ is in $\parami{\P}$. In \cite{grohe} it is shown that this membership is witnessed by a computable sequence of algorithms known as $k$-dimensional Weisfeiler-Lehman algorithm $W_k$ with $k \in \N$. This means for every proper minor-closed graph class $\mathcal{C}$ there exists an $i \in \N$ such that isomorphism on graphs from $\mathcal{C}$ is correctly decided by $W_i$. It holds that $W_{k+1}$ correctly decides isomorphism for a pair of graphs whenever $W_k$ does. This means the sequence of algorithms $(W_k)_{k \in \N}$ is strongly uniform. Therefore $(\mathrm{GI},\minorfree{\mathcal{G}} \times \minorfree{\mathcal{G}})$ is in  $\paramui{\P}$ due to Lemma~\ref{lem:usc}. 
The runtime of $W_k$  is $n^{\mathcal{O}(k)}$. Therefore $(\mathrm{GI},\minorfree{\mathcal{G}} \times \minorfree{\mathcal{G}})$ is in $\paramsui{\P}$ via $(W_k)_{k \in \N}$ and $f(k) = ck$ for some $c \in \N$.  

In \cite{alon} the parameterized complexity of the dominating set problem is considered.
Let $\mathrm{DS} = \set{ (G,k) \in \mathcal{G} \times \N }{ G \text{ has a dominating set $S$ with $|S| \leq k$} }$. They give an algorithm $A$ which gets a graph $G$ and $k,d \in  \N$ as input and outputs a dominating set of size at most $k$ if it exists under the promise that $G$ has degeneracy at most $d$. This algorithm runs in time $k^{\mathcal{O}(dk)} \cdot n$. The algorithm $A$ can also be seen as a sequence of algorithms $(A_{l,d})_{l,d \in \N}$ where the algorithm $A_{l,d}$ gets $(G,k) \in \mathcal{G} \times \N$ as input and outputs a dominating set of size at most $k$ under the promise that $G$ has degeneracy at most $d$ and $k \leq l$. Due to our previous remark it follows that this sequence of algorithms can be made strongly uniform. The algorithm $A_{l,d}$ runs in time $l^{cdl} \cdot n$ for some constant $c \in \N$. This means $(\mathrm{DS},\text{U.S.} \times \tpfin)$ is in $\paramsui{\TIME(n)}$ (and thus $\FPT'$) via $(A_{j,j})_{j \in \N}$ and $f(j) = j^{cj^2}$ where U.S. means uniformly sparse.  

Lower bounds from classical complexity can be easily connected to the parameterized interpretation.
If there is a promise problem $(L,P)$ for which there are reasons to believe that it is not in $\P$ (such as being $\NP$-hard) then for all parameterizations $\mathbb{P}$ with $P \in \mathbb{P}$ the parameterized problem $(L,\mathbb{P})$ can be believed to not be in $\parami{\P}$ for the same reasons.
For instance, the graph isomorphism problem for graphs with degeneracy at most two is $\GI$-complete. Therefore showing that $(\mathrm{GI},\text{U.S.} \times \text{U.S.})$ is in $\parami{P}$ is at least as hard as showing that $\mathrm{GI}$ is in $\P$ since the class of graphs with degeneracy at most two is uniformly sparse.
In fact, these statements are equivalent, i.e.~$(\mathrm{GI},\text{U.S.} \times \text{U.S.})$ is in $\parami{P}$ iff $\mathrm{GI}$ is in $\P$.

\subsection{Reductions}
\begin{definition}[Promise Reduction]
    Let $\ccex$ be a complexity class. We define $\amred$-reductions on promise problems as follows. Let $(L,P)$ and $(L',P')$ be promise problems over alphabets $\Sigma$ and $\Delta$, respectively.
    We say $(L,P) \amred (L',P')$ if there exists a total function $r \colon \Sigma^* \rightarrow \Delta^*$ such that:
    \begin{itemize}
        \item for all $x \in P$ it holds that $x \in L \Leftrightarrow r(x) \in L'$
        \item for all $x \in P$ it holds that $r(x) \in P'$
        \item the function $r$ is computable by an algorithm from $\ccex$
    \end{itemize}
\end{definition}

The first condition says that the reduction function $r$ has to correctly translate all instances which are in the promise $P$. The second condition states that every translated instance from $P$ must be in $P'$, i.e.~$r(P) = \set{r(x)}{x \in P} \subseteq P'$. 

\begin{definition}[Non-Uniform Reduction]
    Let $\ccex$ be a complexity class and let $(L,\mathbb{P})$ and $(L',\mathbb{P}')$ be parameterized problems. We say  $(L,\mathbb{P}) \amred (L',\mathbb{P}')$ if for all $P \in \mathbb{P}$ there exists a $P' \in \mathbb{P}'$ such that $(L,P) \amred (L',P')$.
\end{definition}

It is straightforward to verify that $\pmred$ is reflexive and transitive on parameterized problems. It inherits these properties from its interpretation on promise problems. Also, $\parami{\P}$ is closed under it, i.e.~$(L,\mathbb{P}) \pmred (L',\mathbb{P}')$ and $(L',\mathbb{P}')$ in $\parami{\P}$ implies $(L,\mathbb{P})$ in $\parami{\P}$.

\begin{definition}[Uniform Reduction]
	Let $\ccex$ be a complexity class and let $(L,\mathbb{P})$ and $(L',\mathbb{P}')$ be parameterized problems over  alphabets $\Sigma$ and $\Delta$, respectively. We say $(L,\mathbb{P}) \amredu (L',\mathbb{P}')$ if there exists a computable sequence of algorithms $(M_k)_{k \in \N}$ from $\ccex$ such that $M_i$ computes a total function $r_i \colon \Sigma^* \rightarrow \Delta^*$ for all $i \in \N$ and a parameter $\kappa$ over $\Sigma$ such that:
	\begin{enumerate}
		\item   $\mathbb{P} \subseteq \mathbb{K}(\kappa)$ 
		\item 	for all $i \in \N$ it holds that $(L,\kappa_i) \amred (L',r_i(\kappa_i))$ via $r_i$ and $r_i(\kappa_i) \in \mathbb{P}'$
        \item  	for all $i \in \N$ and $x \in \Sigma^*$ it holds that $\kappa(x) \leq i$ implies $r_{\kappa(x)}(x) = r_i(x)$     
	\end{enumerate}
	\label{def:urn}
\end{definition}

The third condition can be interpreted as the requirement that a reduction function $r_j$ maintains the translation of a previous reduction function $r_i$ ($i < j$) if the instance was already guaranteed to be correctly translated by $r_i$, i.e.~$x \in \kappa_i$. 

\begin{lemma}
    $\pmredu$ is transitive.  
    \label{lem:trans}  
\end{lemma}
\begin{proof}[Proof of Lemma~\ref{lem:trans}]
    Let $(L,\mathbb{P}), (L',\mathbb{P}'),  (L'',\mathbb{P}'')$ be parameterized problems over alphabets $\Sigma,\Delta,\Gamma$, respectively. 
    Let $(L,\mathbb{P}) \pmredu (L',\mathbb{P}')$ via $(M_k)_{k \in \N}$ and parameter $\kappa$ over $\Sigma$; $M_i$ computes the total function $r_i \colon \Sigma^* \rightarrow \Delta^*$ for all $i \in \N$. 
    Let $(L',\mathbb{P}') \pmredu (L'',\mathbb{P}'')$ via $(M'_k)_{k \in \N}$ and parameter $\kappa'$ over $\Delta$; $M'_i$ computes the total function $r'_i \colon \Delta^* \rightarrow \Gamma^*$ for all $i \in \N$.     
    For $i \in \N$ let $M_i''$ be the DTM which on input $x \in \Sigma^*$ computes $r_i'(r_i(x))$ by running $M_i$ and $M_i'$.  For $x \in \Sigma^*$ let $\kappa''(x) = \max \{ \kappa(x), \kappa'(r_{\kappa(x)}(x)) \} $. Let $r''_i(x) = r'_i(r_i(x))$ for all $i \in \N$. We claim that $(L',\mathbb{P}') \pmredu (L'',\mathbb{P}'')$ via $(M''_k)_{k \in \N}$ and $\kappa''$.  
    
    First, we argue that $\mathbb{P} \subseteq \mathbb{K}(\kappa'')$. Let $P \in \mathbb{P}$. There exists an $i \in \N$ such that $P \subseteq \kappa_i$. This means $\kappa(x) \leq i$ for all $x \in P$. 
    It holds that $r_i(\kappa_i) \in \mathbb{P}'$. Therefore there exists a $j \in \N$ such that $r_i(\kappa_i) \subseteq \kappa'_j$.  This means $\kappa'(r_i(x)) \leq j$ for all $x \in P$. Due to the third condition of Definition~\ref{def:urn} it holds that $r_{\kappa(x)}(x) = r_i(x)$ for all $x \in P$. This means for all $x \in P$ it holds that $\kappa'(r_{\kappa(x)}(x)) \leq j$ and thus $\kappa''(x) \leq \max\{i,j\}$, i.e.~$P \subseteq \kappa''_l$ with $l = \max\{i,j\}$.
    
    Secondly, we argue that for all $i \in \N$ it holds that (1) $(L,\kappa''_i) \amred (L'',r''_i(\kappa''_i))$ via $r''_i$ and (2) $r''_i(\kappa''_i) \in \mathbb{P}''$. Let $i \in \N$. For (1) we show $x \in L \Leftrightarrow r''_i(x) \in L''$ for all $x \in \kappa''_i$. Let $x \in \kappa''_i$, which implies $\kappa(x) \leq i$. It follows that $x' := r_{\kappa(x)}(x) = r_i(x)$ and thus $x \in L \Leftrightarrow x' \in L'$. It holds that $\kappa'(x') \leq i$. This implies $x'' := r'_{\kappa'(x')}(x') =  r'_{i}(x')$ and therefore $x' \in L' \Leftrightarrow x'' \in L''$. It holds that $x'' = r'_i(x') = r'_i(r_i(x)) = r''_i(x)$ which concludes our claim. For (2) we show that $r''_i(\kappa''_i) \in \mathbb{P}''$. It holds that $r_i(\kappa_i) \in \mathbb{P}'$. This means there exists a $j \in \N$ such that $r_i(\kappa_i) \subseteq \kappa'_j$.
    We assume w.l.o.g.~that $j \geq i$. 
    It holds that $r'_j(\kappa'_j) \in \mathbb{P}''$.  We claim that $r''_i(\kappa''_i) \subseteq r'_j(\kappa'_j)$, which implies $r''_i(\kappa''_i) \in \mathbb{P}$. Let $x'' \in r''_i(\kappa''_i)$. This means there exists an $x \in \kappa''_i$ such that $x'' = r''_i(x)$. It holds that $x \in \kappa_i$ and thus $r_i(x) \in \kappa'_j$ and therefore $r'_j(r_i(x)) \in r'_j(\kappa'_j)$. It holds that $x'' = r'_i(r_i(x)) = r_j'(r_i(x))$ since $\kappa'(r_i(x)) \leq i \leq j$, which concludes our claim.   
    
    Thirdly, we show that for all $i \in \N$ and $x \in \Sigma^*$ with $\kappa''(x) \leq i$ it holds that $r''_{\kappa''(x)}(x) = r''_i(x)$. It holds that $x' := r_{\kappa(x)}(x) = r_{\kappa''(x)}(x) = r_i(x)$ because $\kappa(x) \leq \kappa''(x) \leq i$. It remains to argue that $r'_i(x') = r'_{\kappa''(x)}(x')$. It holds that $\kappa'(x') \leq \kappa''(x)$ and therefore $r'_{\kappa'(x')}(x') = r'_{\kappa''(x)}(x') = r'_i(x')$.
\end{proof}

\begin{lemma}
    $\paramui{\P}$ is closed under $\pmredu$. 
    \label{lem:uclosure}
\end{lemma}
\begin{proof}
 Let $(L,\mathbb{P})$ and $(L',\mathbb{P}')$ be parameterized problems over alphabets $\Sigma$ and $\Delta$, respectively. Let $(L,\mathbb{P}) \pmredu (L',\mathbb{P}')$ via $(M_k)_{k \in \N}, \kappa$ and let $r_i \colon \Sigma^* \rightarrow \Delta^*$ be the function computed by $M_i$ for $i \in \N$. Let $(L',\mathbb{P}')$ be in $\paramui{\P}$ via $(M'_k)_{k \in \N}, \kappa'$.   
For $k \in \N$ let the DTM $M^*_k$ be defined as follows. On input $x \in \Sigma^*$ it computes $x' = r_k(x)$ by running $M_k$ on $x$ and then runs $M'_k$ on $x'$. Let $\kappa^*(x) = \max \{ \kappa(x) , \kappa'(r_{\kappa(x)}(x) ) \}$ for $x \in \Sigma^*$.
    We claim that $(L,\mathbb{P})$ is in $\paramui{\P}$ via $(M^*_k)_{k \in \N}, \kappa^*$.
    The sequence $(M^*_k)_{k \in \N}$ and $\kappa^*$ are both computable. The runtime of $M^*_k$ is polynomial since it is the sum of the runtime of $M_k$ and $M'_k$.  
    
   We show that $L \cap \kappa^*_i = L(M^*_i) \cap \kappa^*_i$ holds for all $i \in \N$. Let $x \in \kappa^*_i$. This implies $\kappa(x) \leq i$ yand $\kappa'(r_{\kappa(x)}(x)) \leq i$. If we run $M^*_i$ on $x$ it first computes $r_i(x)$ and then runs $M'_i$ on $r_i(x)$. Since $x \in \kappa_i$ it holds that $r_i$ correctly translates $x$, i.e.~$x \in L \Leftrightarrow r_i(x) \in L'$. Moreover, $M'_i$ correctly decides $r_i(x)$ if $\kappa'(r_i(x)) \leq i$. This premise holds because $r_i(x) = r_{\kappa(x)}(x)$ due to the third condition of Definition~\ref{def:urn}. 
   
   It remains to argue that $\mathbb{P} \subseteq \mathbb{K}(\kappa^*)$. 
   Let $P \in \mathbb{P}$. There exists an $a \in \N$ such that $P \subseteq \kappa_a$. It holds that $(L,\kappa_a)$ reduces to $(L',r_a(\kappa_a))$ and $r_a(\kappa_a) \in \mathbb{P}'$. There exists a $b \in \N$ such that $r_a(\kappa_a) \subseteq \kappa'_b$. Let $c = \max  \{a,b\}$. We claim that $P \subseteq \kappa^*_{c}$. Let $x \in P$. 
   %It holds that $\kappa^*(x) = \max \{ \kappa(x) , \kappa'(r_{\kappa(x)}(x)) \}$.
   Since $x \in \kappa_a$ it holds that $\kappa(x) \leq a \leq c$. This implies $r_{\kappa(x)}(x) = r_a(x)$ because $\kappa(x) \leq a$. Additionally, $r_a(x) \in \kappa'_b$ which means $\kappa'(r_a(x)) \leq b \leq c$. Therefore $\kappa''(x) \leq \max\{a,b\} \leq c$. 
\end{proof}

For the sake of brevity we define strongly uniform reductions only for the classes $\TIME(n^c)$. 

\begin{definition}[Strongly Uniform Reduction]
	Let $(L,\mathbb{P})$ and $(L',\mathbb{P}')$ be parameterized problems and let $c \in \N$. We say $(L,\mathbb{P}) \timecredsu (L',\mathbb{P}')$ if there exists 	a computable sequence of algorithms $(M_k)_{k \in \N}$ from $\TIME(n^c)$, a parameter $\kappa$ over $\Sigma$ and a computable, total function $f \colon \N \rightarrow \N$ such that:
	\begin{itemize}
		\item $(L,\mathbb{P}) \timecredu (L',\mathbb{P}')$ via $(M_k)_{k \in \N}, \kappa$
		\item for all $k \in \N$ and $x \in \Sigma^*$  it holds that $M_k$ on $x$ runs at most $f(k) \cdot {|x|}^c$ steps
	\end{itemize}
\end{definition}

We say $(L,\mathbb{P}) \mysufptred (L',\mathbb{P}')$  if there exists a $c \in \N$ such that $(L,\mathbb{P}) \timecredsu (L',\mathbb{P}')$.

\begin{lemma}
    $\FPT'$ is closed under $\mysufptred$. 
\end{lemma}
\begin{proof}
    Let $(L,\mathbb{P})$ and $(L',\mathbb{P}')$ be parameterized problems over alphabets $\Sigma$ and $\Delta$, respectively. 
    Let $(L,\mathbb{P}) \mysufptred (L',\mathbb{P}')$ via $(M_k)_{k \in \N}, \kappa, f$ and $c \in \N$. This means for all $k \in \N$ the DTM $M_k$ runs at most $f(k) \cdot n^c$ steps on all inputs of length $n$.  
    Let $(L',\mathbb{P}')$ be in $\FPT'$ via $(M'_k)_{k \in \N}, \kappa', f'$ and $d \in \N$. For all $k \in \N$ the DTM $M'_k$ runs at most $f'(k) \cdot n^{d}$ steps on all inputs of length $n$.
    
    We claim that $(L,\mathbb{P})$ is in $\FPT'$ via $(M^*_k)_{k \in \N}$, $\kappa^*$, $f^*(k) = f(k) \cdot f'(k)$ and $c+d$ where $M^*_k$ and $\kappa^*$ are defined as in the proof of Lemma~\ref{lem:uclosure}. For all $k \in \N$ it holds that $M^*_k$ runs at most $f(k) \cdot n^c + f'(k) \cdot n^d \leq f^*(k) \cdot n^{c+d}$ steps on all inputs of length $n$.
    The remaining conditions are satisfied due to the same argument given in the proof of Lemma~\ref{lem:uclosure}.
\end{proof}

\section{Relation to Existing Formalization}
\label{sec:connection}
We use Flum and Grohe's formalization to connect the definitions in the previous section to the existing theory. First, we show how our definition of a parameterized problem relates to the one given by Flum and Grohe. Let us call a total function $\kappa \colon \Sigma^* \rightarrow \N$ a n.n.c.~parameter (not necessarily computable) for an alphabet $\Sigma$.

\begin{theorem}
    Let $\mathbb{P}$ be a set of languages over an alphabet $\Sigma$. It holds that $\mathbb{P}$ is a parameterization iff there exists a n.n.c.~parameter $\kappa$ over $\Sigma$ such that $\mathbb{P} = \mathbb{K}(\kappa)$. 
\end{theorem}
\begin{proof}    
    ``$\Rightarrow$'': Let $\mathbb{P}$ be a parameterization.  We define a n.n.c.~parameter $\kappa$ over $\Sigma$ such that $\mathbb{P} = \mathbb{K}(\kappa)$. Let $\{P_1,P_2,\dots \}$ be a countable subset of $\mathbb{P}$ whose closure under subsets equals $\mathbb{P}$. Let $P'_c = \bigcup_{i=1}^c P_i$. It holds that $P'_c$ is in $\mathbb{P}$ for every $c \in \N$ because $\mathbb{P}$ is closed under union. Define $\kappa(x)$ as the least $k$ such that $x \in P'_k$.

    ``$\Leftarrow$'': Let $\kappa$ be a n.n.c.~parameter over $\Sigma$ such that $\mathbb{P} = \mathbb{K}(\kappa)$. Let $P,P'$ be languages over $\Sigma$ which are both bounded by $\kappa$. This means $P \subseteq \kappa_i$ and $P' \subseteq \kappa_j$ for some $i,j \in \N$. We assume w.l.o.g.~that $i \leq j$ and therefore $L \cup L' \subseteq \kappa_j$. Therefore $\mathbb{K}(\kappa)$ is closed under union. Since $\kappa$ is total it follows that $\{x\}$ is in $\mathbb{K}(\kappa)$ for all $x \in \Sigma^*$. A countable subset of $\mathbb{K}(\kappa)$ such that its closure under subsets equals $\mathbb{K}(\kappa)$ is given by $\{\kappa_1, \kappa_2, \dots \} $.        
\end{proof}

Stated differently, a parameterization is a set of languages that is bounded by some (possibly uncomputable) parameter. This means instead of $(L,\kappa)$ we have been looking at $(L,\mathbb{K}(\kappa))$.
We still have to argue that defining the complexity in terms of $(L,\mathbb{K}(\kappa))$ is the same as defining it in terms of $(L,\kappa)$. More precisely, could there be two parameters $\kappa, \tau$ with $\mathbb{K}(\kappa) = \mathbb{K}(\tau) = \mathbb{P}$ such that the complexity of $(L,\kappa)$ differs from the one of $(L,\tau)$? This would make it ill-defined to speak about the complexity of $(L,\mathbb{P})$. If one considers the complexity of $(L,\kappa)$ to be its closure under uniform fpt-reductions then this is not the case. For two n.n.c.~parameters $\kappa,\tau$ over $\Sigma$ let us say $\kappa \preccurlyeq \tau$ if there exists a function $f \colon \N \rightarrow \N$ such that $\kappa(x) \leq f(\tau(x))$ holds for all $x \in \Sigma^*$. We say $\kappa$ and $\tau$ are equivalent if $\kappa \preccurlyeq \tau$ and $\tau \preccurlyeq \kappa$.
\begin{fact}        
    Let $\kappa,\tau$ be n.n.c.~parameters over the alphabet $\Sigma$. It holds that $\kappa \preccurlyeq \tau$ iff $\mathbb{K}(\tau) \subseteq \mathbb{K}(\kappa)$. 
\end{fact}
\begin{proof}        
    ``$\Rightarrow$'': Let $\kappa \preccurlyeq \tau$ via a monotone function $f \colon \N \rightarrow \N$, i.e.~$\kappa(x) \leq f(\tau(x))$ for all $x \in \Sigma^*$. We show inductively that for every $i \in \N$ it holds that $\tau_i \subseteq \kappa_{f(i)}$, which implies $\mathbb{K}(\tau) \subseteq \mathbb{K}(\kappa)$. For the base case $i = 1$ it holds that $x \in \tau_1$ iff $\tau(x) = 1$. It follows that $\kappa(x) \leq f(1)$ and therefore $x \in \kappa_{f(1)}$. For the inductive step $i \rightarrow i + 1$ it must be the case that $x$ is either in $\tau_{i+1} \setminus \tau_i$ or in $\tau_i$. If $x$ is in $\tau_i$ then by induction hypothesis it holds that $x \in \kappa_{f(i)}$. Since $f$ is monotone it follows that $x \in \kappa_{f(i+1)}$ as well. For the other case it holds that $\tau(x) = i+1$ and therefore $\kappa(x) \leq f(i+1)$ which means $x \in \kappa_{f(i+1)}$.      
    
    ``$\Leftarrow$'': Since $\mathbb{K}(\tau) \subseteq \mathbb{K}(\kappa)$ there exists a function $f \colon \N \rightarrow \N$ such that $\tau_i \subseteq \kappa_{f(i)}$ for all $i \in \N$. We argue that $\kappa(x) \leq f(\tau(x))$ for all $x \in \Sigma^*$. Let $\tau(x) = k$ for some $k \in \N$. Then it holds that $x \in \tau_k$ and therefore $x \in \kappa_{f(k)}$ as well. This means $\kappa(x) \leq f(k) = f(\tau(x))$.
\end{proof}

Thus two parameters are equivalent iff they bound the same set of languages. Observe that if $\kappa$ and $\tau$ are equivalent then $(L,\kappa)$ and $(L,\tau)$ are equivalent w.r.t.~uniform fpt-reductions (by that we mean \cite[Def.~2.1]{flum} without the requirement that the functions $f$ and $g$ have to be computable) for all languages $L$. 
Therefore our definition of a parameterized problem seamlessly fits into the existing theory. 
One can think of a parameter as a representation of a parameterization. For example, it is more natural to define the combination of two parameters $\kappa, \tau$ over the same alphabet in terms of the parameterizations which they represent: 
$$  \mathbb{K}(\kappa + \tau) = \mathbb{K}(\kappa \cdot \tau) = \mathbb{K}(\max \{\kappa , \tau \}) = \mathbb{K}(\kappa) \cap \mathbb{K}(\tau)  $$

Next, we show how slices and promise problems are related. Given a parameterized problem $(L,\kappa)$ and $i \in \N$. Let $P = \set{x}{\kappa(x) = i}$. The $i$-th slice of $(L,\kappa)$ is in $\P$ iff  $\exists L' \in \P : L \cap P = L'$ whereas the promise problem $(L,P)$ is in $\P$ iff $\exists L' \in \P : L \cap P = L' \cap P$. The lack of the second `$\cap P$' in the case of slices is what leads to the conflation of problem and parameter complexity. Also, it shows that membership of a slice implies membership of the corresponding promise problem. 

\begin{fact}
    Let $(L,\kappa)$ be a parameterized problem. If $(L,\kappa)$ is in $\XP_{\mathrm{nu}}$ then $(L,\mathbb{K}(\kappa))$ is in $\parami{\P}$. 
\end{fact}
\begin{proof}
    Assume that $(L,\kappa)$ is in $\XP_{\mathrm{nu}}$. We show that for every $P \in \mathbb{K}(\kappa)$ it holds that $(L,P)$ is in $\P$. Let $P \in \mathbb{K}(\kappa)$. There exists a $c \in \N$ such that $P \subseteq \kappa_c$. Let $L' = \cup_{i=1}^c (L,\kappa)_i$. Since every slice of $(L,\kappa)$ is in $\P$ and $\P$ is closed under union it follows that $L'$ is in $\P$.  It holds that $L \cap P =  L' \cap P$. 
\end{proof}

The class $\P$ in the above statement can be replaced by any complexity class closed under union. If $\kappa$ is  polynomial-time computable then the converse direction holds as well, i.e.~$(L,\mathbb{K}(\kappa))$ in $\parami{\P}$ implies $(L,\kappa)$ in $\XP_{\mathrm{nu}}$. 

\begin{fact}
    Let $(L,\kappa)$ be a parameterized problem such that $\kappa$ is polynomial-time computable. If $(L,\kappa)$ is in $\FPT$ then $(L,\mathbb{K}(\kappa))$ is in $\FPT'$. 
\end{fact}
\begin{proof}
    Assume that $(L,\kappa)$ is in $\FPT$ and $\kappa$ is computable in time $n^d$ for some $d \in \N$. This means there exists a DTM $M$, $c \in \N$ and a computable function $f \colon \N \rightarrow \N$ such that $M$ correctly decides $L$ and runs in time $f(\kappa(x)) \cdot {|x|}^c$ for all instances $x$. We assume w.l.o.g.~that $f$ is monotone.
    For $i \in \N$ let the DTM $M_i$ be defined as follows. On input $x$ it checks if $\kappa(x) \leq i$. If this is the case then it runs $M$ on $x$, otherwise it rejects.    
    It holds that $(L,\mathbb{K}(\kappa))$ is in $\paramsui{\TIME(n^{c+d})}$ via $(M_k)_{k \in \N}$, $\kappa$ and $f$. For all $k \in \N$ and all inputs $x$ it holds that $M_k$ on $x$ runs in time  $n^d + f(k) \cdot n^c$. 
\end{proof}

If we strengthen the requirement in Definition~\ref{def:upc} that $\kappa$ has to be polynomial-time computable then the two classes coincide, i.e.~if a parameterized problem $(L,\mathbb{P})$ is in $\FPT'$ via $(M_k)_{k \in \N}, \kappa, f$ then $(L,\kappa)$ is in $\FPT$ via $M$ where $M$ is defined as follows. On input $x$ the DTM $M$ computes $k = \kappa(x)$, then $M_k$ and finally it runs $M_k$ on $x$. 

\begin{fact}
    Let $(L,\kappa)$ and $(L',\kappa')$ be parameterized problems and $\kappa$ is polynomial-time computable. If $(L,\kappa) \fptred (L',\kappa')$ 
    then $(L,\mathbb{K}(\kappa)) \mysufptred (L',\mathbb{K}(\kappa'))$. 
\end{fact}
\begin{proof}
    Assume $(L,\kappa)$ and $(L',\kappa')$ are over the alphabets $\Sigma$ and $\Delta$, respectively. 
    Let $(L,\kappa) \fptred (L',\kappa')$. This means there exists a DTM $M$, $c \in \N$ and computable, total functions $f,g \colon \N \rightarrow \N$ such that  $M$ computes a reduction function $R \colon \Sigma^* \rightarrow \Delta^*$ ($x \in L \Leftrightarrow R(x) \in L'$) and runs in time $f(\kappa(x)) \cdot {|x|}^c$  and  $\kappa'(R(x)) \leq g(\kappa(x)) $ for all $x \in \Sigma^*$. We assume w.l.o.g.~that $f$ and $g$ are monotone and $\kappa$ is computable in time $n^c$.
    For $i \in \N$ let the DTM $M_i$ be defined as follows. On input $x$ it checks if $\kappa(x) \leq i$. If this is the case then it runs $M$ on $x$, otherwise it outputs the empty word. For $i \in \N$ let $r_i \colon \Sigma^* \rightarrow \Delta^*$ denote the reduction function computed by $M_i$.
    We claim that  $(L,\mathbb{K}(\kappa)) \timecredsu (L',\mathbb{K}(\kappa'))$ via $(M_k)_{k \in \N}, \kappa, f$. For all $k \in \N$ and $x \in \Sigma^*$ the DTM $M_k$ runs in time $(f(k)+1) \cdot {|x|}^c$. It remains to argue that $(L,\mathbb{K}(\kappa)) \timecredu (L',\mathbb{K}(\kappa'))$ via $(M_k)_{k \in \N}, \kappa$. 
    The first and third condition of Definition~\ref{def:urn} are simple to verify. For the second condition we have to show that (1) $(L,\kappa_i) \timecred (L,r_i(\kappa_i))$ via $r_i$ and (2) $r_i(\kappa_i)\in \mathbb{K}(\kappa')$ holds for all $i \in \N$. Let $x \in \kappa_i$. Part (1) holds because $r_i(x) = R(x)$ for all $x \in \kappa_i$ and thus $r_i$ is a correct reduction. Part (2) holds because $\kappa'(r_i(x)) = \kappa'(R(x)) \leq g(\kappa(x)) \leq g(i)$ for all $x \in \kappa_i$ and therefore $r_i(\kappa_i) \subseteq \kappa'_{g(i)}$.
\end{proof}

In a nutshell, the existing classes and reductions in parameterized complexity can be seen as a special case of our definitions where the parameter $\kappa$ in Definition~\ref{def:upc} and \ref{def:urn} is required to be polynomial-time computable.

\section{Conclusion}
\label{sec:final}
The premise of our formalization is that the parameter value should not be part of the input, which we justify as follows. A parameterized problem consists of the three components input, parameter and output. The definition of a parameterized problem as subset of $\Sigma^* \times \N$ conflates input and parameter. It is not clear why this conflation is acceptable or even necessary. Given an arbitrary parameterized problem $L \subseteq \Sigma^* \times \N$, it is unclear what membership in $L$ intuitively means. This is directly related to the question at the end of Section~\ref{ss:df} whether every subset of $\Sigma^* \times \N$ corresponds to a natural parameterized problem. A consequence of this conflation is that the structure of parameterizations becomes invisible. For example, one cannot see that the set of parameterizations forms a bounded lattice. Also, having to add the parameter value to the input can be problematic when studying more restricted models of computation such as in the context of descriptive complexity.

The obvious alternative is to compute the parameter value from the input if it is not part thereof. The problem with this approach is that suddenly the complexity of computing the parameter value becomes part of the complexity of the parameterized algorithm. In practice, this is not the case. For example, if one designs an algorithm for some problem parameterized by tree-width then computing the tree-width is certainly not part of the parameterized algorithm; the algorithm designer expects the tree-width of the input to be given. If one considers a parameterized algorithm to be a single algorithm then this leaves no room but to make the parameter value part of the input. However, we have argued in the beginning of Section~\ref{sec:nar} that a parameterized algorithm can be seen as a sequence of algorithms. From this perspective, the algorithm designer must find an algorithm $A_k$ for every $k \in \N$ such that $A_k$ solves the problem for all inputs where the tree-width is at most $k$ (a promise problem). In conjunction with the observation that the complexity of a parameterized problem $(L,\kappa)$ actually depends on $(L,\mathbb{K}(\kappa))$, this leads to the formalization presented here. 

Our definitions naturally translate to other kinds of complexity and thus provide us with a unified formal notion of parameterized complexity, which only requires formalizing promise problems and reductions among them. In turn, one immediately obtains parameterized classes and reductions. 
For example, it should be unambiguously clear what $\parami{\APX}, \parami{\ComplexityFont{IncP}}$ and $\paramui{\ComplexityFont{\#P}}$ mean, assuming the base classes are known. More generally, any complexity-theoretic notion which can be sensibly defined on a restricted set of inputs can be interpreted in a parameterized context. This is in contrast to slices, which do not generalize beyond decision problems and also conflate problem and parameter complexity. The view that slices are inappropriate is also shared by Goldreich who states that ``[Slices] miss the true nature of the original computational problem [...]'' and ``the conceptually correct perspective [...] is a promise problem [...]'' \cite[p.~255f]{goldreich}.

We consider the application of a parameterized algorithm $(A_k)_{k \in \N}$ to be a two-step process. Given an input $x$ (or a set of inputs) the first step is to determine which algorithm $A_i$ to use (algorithm selection). 
In our definition of a uniform parameterized complexity class a parameterized algorithm is accompanied by a parameter $\kappa$ which allows us to determine this information: $x$ is guaranteed to be correctly solved by $A_i$ for all $i \geq \kappa(x)$. 
The second step is to run the selected algorithm on $x$ (algorithm execution). Unlike in the definition of $\FPT$ we only require the algorithm execution to `run in fpt-time'. This means the complexity of algorithm selection, i.e.~computing an upper bound on $\kappa(x)$, is not part of the parameterized complexity in our framework. This follows the generally recognized principle of separation of concerns and also accounts for the fact that the algorithm selection step can profit from domain-specific knowledge which is not available in theory. 

In summary, our formalization tightly corresponds to the intuitive notion of a parameterized problem, does not require any arbitrary choices (such as polynomial-time computable parameters) and neatly generalizes to other kinds of complexity. Also, rather than looking at classical and parameterized complexity as two incomparable entities our formalization shows how classical complexity can be formally interpreted as a special case of parameterized complexity.

Is there a natural parameterized algorithm $(A_k)_{k \in \N}$ for some natural problem $X$ which is not strongly monotone w.r.t.~$X$? The problem $X$ can probably be expected to be outside of $\NP$ (see the second paragraph after Lemma~\ref{lem:usc}). 

\subparagraph*{Acknowledgements.} We thank the anonymous reviewer who pointed out that there is a generic translation from Downey and Fellows' formalization of a parameterized problem to Flum and Grohe's which preserves fixed-parameter tractability.

\printbibliography[heading=bibintoc]

\end{document}